\documentclass[11pt]{article}
\usepackage[a4paper]{geometry}
\usepackage{amsfonts, amsmath, amssymb, amsthm, graphicx, caption, authblk, url, multirow, makecell, framed, float, xcolor, enumitem, tikz}
\theoremstyle{definition}

\newtheorem{lemma}{Lemma}

\theoremstyle{remark}
\newtheorem*{remark}{Remark}

\newcommand*{\mybox}[1]{%
  \framebox{\raisebox{0cm}[0.5\baselineskip][0.05\baselineskip]{%
    \hbox to 0.1cm {\hss#1\hss}}}\hspace{0.05cm}}

\begin{document}
\title{Physical Zero-Knowledge Proof for Numberlink Puzzle and $k$ Vertex-Disjoint Paths Problem\thanks{A preliminary version of this paper \cite{ruangwises} has appeared in the proceedings of FUN 2021.}}
\author[1]{Suthee Ruangwises\thanks{\texttt{ruangwises.s.aa@m.titech.ac.jp}}}
\author[1]{Toshiya Itoh\thanks{\texttt{titoh@c.titech.ac.jp}}}
\affil[1]{Department of Mathematical and Computing Science, Tokyo Institute of Technology, Tokyo, Japan}
\date{}
\maketitle

\begin{abstract}
Numberlink is a logic puzzle with an objective to connect all pairs of cells with the same number by non-crossing paths in a rectangular grid. In this paper, we propose a physical protocol of zero-knowledge proof for Numberlink using a deck of cards, which allows a prover to convince a verifier that he/she knows a solution without revealing it. In particular, the protocol shows how to physically count the number of elements in a list that are equal to a given secret value without revealing that value, the positions of elements in the list that are equal to it, or the value of any other element in the list. Finally, we show that our protocol can be modified to verify a solution of the well-known $k$ vertex-disjoint paths problem, both the undirected and directed settings.

\textbf{Keywords:} card-based cryptography, zero-knowledge proof, Numberlink, puzzle, vertex-disjoint paths, graph
\end{abstract}

\section{Introduction}
\textit{Numberlink} is a logic puzzle developed by a Japanese company Nikoli, which is famous for creating many popular logic puzzles including Sudoku, Kakuro, Shikaku, and Hashiwokakero. Recently, the puzzle has become increasingly popular, and a large number of Numberlink mobile apps with different names have been developed \cite{google}.

A Numberlink puzzle consists of a rectangular grid with some cells containing a number. Each number appears exactly twice in the grid. The objective of this puzzle is to connect every pair of cells with the same number by a path that can go from a cell to its horizontally or vertically adjacent cell. Paths cannot cross or share a cell with one another. In the official rule \cite{nikoli}, it is not required that all cells in the grid have to be covered by paths. However, a puzzle is generally considered to be well-designed if it has a unique solution, and all cells are covered by paths in that solution.

\begin{figure}
\begin{center}
\begin{tikzpicture}
\draw[step=0.8cm,color=gray] (0,0) grid (4,4);

\node at (3.6,3.6) {4};
\node at (0.4,2.8) {3};
\node at (1.2,2.8) {1};
\node at (3.6,2.8) {3};
\node at (0.4,2.0) {2};
\node at (2.8,2.0) {4};
\node at (0.4,0.4) {2};
\node at (3.6,0.4) {1};
\end{tikzpicture}
\hspace{1.5cm}
\begin{tikzpicture}
\draw[step=0.8cm,color=gray] (0,0) grid (4,4);

\node at (3.6,3.6) {4};
\node at (0.4,2.8) {3};
\node at (1.2,2.8) {1};
\node at (3.6,2.8) {3};
\node at (0.4,2.0) {2};
\node at (2.8,2.0) {4};
\node at (0.4,0.4) {2};
\node at (3.6,0.4) {1};

\draw[very thick] (1.2,2.6) -- (1.2,0.4);
\draw[very thick] (1.2,0.4) -- (3.4,0.4);
\draw[very thick] (0.4,1.8) -- (0.4,0.6);
\draw[very thick] (0.4,3.0) -- (0.4,3.6);
\draw[very thick] (0.4,3.6) -- (2.0,3.6);
\draw[very thick] (2.0,3.6) -- (2.0,1.2);
\draw[very thick] (2.0,1.2) -- (3.6,1.2);
\draw[very thick] (3.6,1.2) -- (3.6,2.6);
\draw[very thick] (3.4,3.6) -- (2.8,3.6);
\draw[very thick] (2.8,3.6) -- (2.8,2.2);
\end{tikzpicture}
\caption{An example of a Numberlink puzzle (left) and its solution (right)}
\label{fig1}
\end{center}
\end{figure}
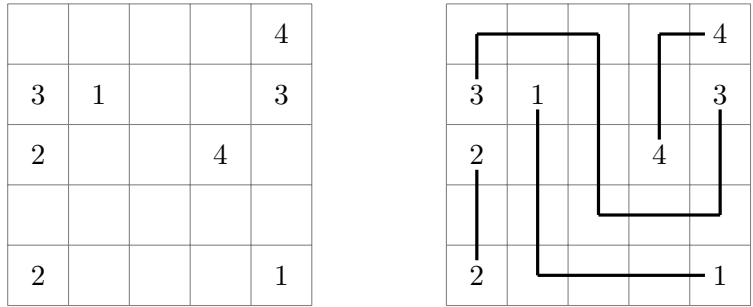

Suppose that Alice, an expert in Numberlink, created a difficult Numberlink puzzle and challenged her friend Bob to solve it. After several tries, Bob could not solve her puzzle. He then claimed that the puzzle does not have a solution and refused to try it anymore. In order to convince Bob that her puzzle actually has a solution without revealing it to him (which would render the challenge pointless), Alice needs some kind of a \textit{zero-knowledge proof}.

\subsection{Zero-Knowledge Proof}
A zero-knowledge proof is an interactive proof between a prover $P$ and a verifier $V$. Both $P$ and $V$ are given an instance $x$ of a computational problem. Only $P$ knows a solution $w$ of $x$, and $V$ cannot obtain $w$ from $x$ by his/her computational power. $P$ wants to convince $V$ that he/she knows a solution of $x$ without revealing any information about $w$ to $V$. A zero-knowledge proof must satisfy the following three properties.
\begin{enumerate}
	\item \textbf{Completeness:} If $P$ knows $w$, then $P$ can convince $V$ with high probability (in this paper, we focus only on the \textit{perfect completeness} property where the probability to convince $V$ is one).
	\item \textbf{Soundness:} If $P$ does not know $w$, then $P$ cannot convince $V$, except with a small probability called \textit{soundness error} (in this paper, we focus only on the \textit{perfect soundness} property where the soundness error is zero).
	\item \textbf{Zero-Knowledge:} $V$ cannot obtain any information about $w$, i.e. there exists a probabilistic polynomial time algorithm $S$ (called the \textit{simulator}) that does not know $w$, and the outputs of $S$ follow the same probability distribution as the outputs of the real protocol.
\end{enumerate}

The concept of zero-knowledge proof was first introduced by Goldwasser et al. \cite{zkp0}. Later, Goldreich et al. \cite{zkp} proved that there exists a zero-knowledge proof for every NP problem. Since Numberlink is known to be NP-complete \cite{np,np2,np1}, one can construct a cryptographic zero-knowledge proof for it. However, such construction requires cryptographic primitives and is neither practical nor intuitive.

Instead, we are interested in constructing a physical protocol using a deck of playing cards. These card-based protocols have benefit that they require only a portable deck of cards which can be found in everyday life, and do not require computers. Moreover, these protocols are easy to understand and verify the security and correctness, even for non-experts, hence they also have a great didactic value.

\subsection{Related Work}
In 2007, Gradwohl et al. \cite{sudoku0} developed the first card-based protocols of zero-knowledge proof for a logic puzzle Sudoku. Each of their several variants of the protocol either uses special scratch-off cards or has a non-zero soundness error. Sasaki et al. \cite{sudoku} later improved the protocol for Sudoku to achieve perfect soundness without using special cards. Other than Sudoku, card-based protocols of zero-knowledge proof for other logic puzzles have been developed as well, including Nonogram \cite{nonogram}, Akari \cite{akari}, Takuzu \cite{akari}, Kakuro \cite{akari,kakuro}, KenKen \cite{akari}, Makaro \cite{makaro}, Norinori \cite{norinori}, and Slitherlink \cite{slitherlink}.

These protocols of zero-knowledge proof employ methods to physically verify specific functions. For example, the protocol for Sudoku \cite{sudoku0} shows how to verify the presence of all numbers in a list without revealing their order, the protocol for Makaro \cite{makaro} shows how to verify that a number is the largest one in a list without revealing any value in the list, and the protocol for Norinori \cite{norinori} shows how to verify the presence of a given number in a list without revealing its position or any other value in the list.

\subsection{Our Contribution}
In this paper, we propose a physical protocol of zero-knowledge proof for Numberlink using a deck of cards. The protocol achieves perfect completeness and perfect soundness properties.

More importantly, by developing the protocol for Numberlink, we also extend the set of functions that are known to be physically verifiable. In particular, our protocol shows how to physically count the number of elements in a list that are equal to a given secret value without revealing that value, the positions of elements in the list that are equal to it, or the value of any other element in the list.

Finally, we show that our protocol can be modified to verify a solution of the undirected $k$ vertex-disjoint paths problem (\textsc{u}$k$-\textsc{dpp}) and the directed $k$ vertex-disjoint paths problem (\textsc{d}$k$-\textsc{dpp}), which are two well-known problems in algorithmic graph theory.

The main differences from the conference version of this paper \cite{ruangwises} are the improvement of the verification phase of our main protocol, which slightly reduces the number of required cards, and the addition of a new protocol for \textsc{d}$k$-\textsc{dpp}, in which the verification phase is divided into two rounds.

\section{Preliminaries}
\subsection{Numberlink Grid}
We consider a Numberlink grid with size $m \times n$ and has $k$ pairs of numbers $1,2,...,k$ written on some of its cells. We call two cells in the grid \textit{adjacent} if they are horizontally or vertically adjacent. Cells with a number written on them are called \textit{terminal cells}, and other cells are called \textit{non-terminal cells}.

A \textit{path} in a solution of a Numberlink puzzle is a sequence of cells $(c_1,c_2,...,c_t)$ where $c_1$ and $c_t$ are terminal cells with the same numbers written on them and all the other cells are non-terminal cells, with $c_i$ being adjacent to $c_{i+1}$ for every $i=1,2,...,t-1$. Also, a path $(c_1,c_2,...,c_t)$ is called \textit{simple} if there is no $i,j$ such that $j>i+1$ and $c_i$ is adjacent to $c_j$.

A Numberlink puzzle is called \textit{well-designed} if it has a unique solution, and all cells are covered by paths in that solution. Note that if a puzzle is well-designed, then every path in its solution must be simple (otherwise if we have a non-simple path $(c_1,c_2,...,c_t)$ with $c_i$ being adjacent to $c_j$ where $j>i+1$, then we can replace it with a shorter path $(c_1,c_2,...,c_i,c_j,c_{j+1}...,c_t)$, creating an alternative solution).

\subsection{Cards}
We use two types of cards in our protocol: \textit{encoding cards} and \textit{marking cards}. An encoding card has either $\clubsuit$ or $\heartsuit$ on the front side, while a marking card has a positive integer on the front side. All cards have an identical back side.

Define $E_y(x)$ to be a sequence of $y$ encoding cards in a row, with all of them being \mybox{$\clubsuit$} except the $x$-th card from the left being \mybox{$\heartsuit$}, e.g. $E_3(1)$ is \mbox{\mybox{$\heartsuit$}\mybox{$\clubsuit$}\mybox{$\clubsuit$}} and $E_4(3)$ is \mbox{\mybox{$\clubsuit$}\mybox{$\clubsuit$}\mybox{$\heartsuit$}\mybox{$\clubsuit$}}. We use $E_y(x)$ to encode a number $x$ in a situation where the maximum possible number is at most $y$. This encoding rule was first considered by Shinagawa et al. \cite{shinagawa} in the context of using regular $y$-gon cards to encode integers in $\mathbb{Z}/y\mathbb{Z}$.

\subsection{Matrix}
Suppose we have $a$ numbers $x_1,x_2,...,x_a$ that are at most $b$. Each number $x_i$ is encoded by a sequence of cards $E_b(x_i)$. We construct a matrix $D(a,b)$ of cards as follows.

First, construct an $a \times b$ matrix of face-down encoding cards, with the $i$-th topmost row being $E_b(x_i)$. Then, on top of the topmost row of the matrix, place face-down marking cards \mybox{1}, \mybox{2}, ..., \mybox{$b$} from left to right in this order. We call this new row Row 0. Also, to the left of the leftmost column of the matrix, place face-down marking cards \mybox{2}, \mybox{3}, ..., \mybox{$a$} from top to bottom in this order (starting at Row 2). We call this new column Column 0. As a result, $D(a,b)$ becomes an incomplete $(a+1) \times (b+1)$ matrix with two cards at the top-left corner removed (see Fig. \ref{fig2}).

\begin{figure}
\begin{center}
\begin{tikzpicture}
\node at (0,0) {\mybox{?}};
\node at (0.5,0) {\mybox{?}};
\node at (1,0) {\mybox{?}};
\node at (1.5,0) {\mybox{?}};
\node at (2,0) {\mybox{?}};
\node at (2.5,0) {\mybox{?}};
\draw[->] (2.8,0) -- (3.2,0);
\node at (3.8,0) {$E_6(x_5)$};

\node at (0,0.6) {\mybox{?}};
\node at (0.5,0.6) {\mybox{?}};
\node at (1,0.6) {\mybox{?}};
\node at (1.5,0.6) {\mybox{?}};
\node at (2,0.6) {\mybox{?}};
\node at (2.5,0.6) {\mybox{?}};
\draw[->] (2.8,0.6) -- (3.2,0.6);
\node at (3.8,0.6) {$E_6(x_4)$};

\node at (0,1.2) {\mybox{?}};
\node at (0.5,1.2) {\mybox{?}};
\node at (1,1.2) {\mybox{?}};
\node at (1.5,1.2) {\mybox{?}};
\node at (2,1.2) {\mybox{?}};
\node at (2.5,1.2) {\mybox{?}};
\draw[->] (2.8,1.2) -- (3.2,1.2);
\node at (3.8,1.2) {$E_6(x_3)$};

\node at (0,1.8) {\mybox{?}};
\node at (0.5,1.8) {\mybox{?}};
\node at (1,1.8) {\mybox{?}};
\node at (1.5,1.8) {\mybox{?}};
\node at (2,1.8) {\mybox{?}};
\node at (2.5,1.8) {\mybox{?}};
\draw[->] (2.8,1.8) -- (3.2,1.8);
\node at (3.8,1.8) {$E_6(x_2)$};

\node at (0,2.4) {\mybox{?}};
\node at (0.5,2.4) {\mybox{?}};
\node at (1,2.4) {\mybox{?}};
\node at (1.5,2.4) {\mybox{?}};
\node at (2,2.4) {\mybox{?}};
\node at (2.5,2.4) {\mybox{?}};
\draw[->] (2.8,2.4) -- (3.2,2.4);
\node at (3.8,2.4) {$E_6(x_1)$};

\node at (0,3.2) {\mybox{1}};
\node at (0.5,3.2) {\mybox{2}};
\node at (1,3.2) {\mybox{3}};
\node at (1.5,3.2) {\mybox{4}};
\node at (2,3.2) {\mybox{5}};
\node at (2.5,3.2) {\mybox{6}};
\node at (4.6,3.2) {(actually face-down)};

\node at (-0.7,0) {\mybox{5}};
\node at (-0.7,0.6) {\mybox{4}};
\node at (-0.7,1.2) {\mybox{3}};
\node at (-0.7,1.8) {\mybox{2}};
\draw[] (-0.7,-0.3) -- (-0.7,-0.6);
\draw[->] (-0.7,-0.6) -- (-0.2,-0.6);
\node at (1.6,-0.6) {(actually face-down)};

\draw[] (-1,-0.3) -- (-1,4.7);
\draw[] (-2.5,3.6) -- (3.2,3.6);

\node at (-1.3,0) {5};
\node at (-1.3,0.6) {4};
\node at (-1.3,1.2) {3};
\node at (-1.3,1.8) {2};
\node at (-1.3,2.4) {1};
\node at (-1.3,3.2) {0};
\node at (-2,1.5) {Row};

\node at (-0.7,3.9) {0};
\node at (0,3.9) {1};
\node at (0.5,3.9) {2};
\node at (1,3.9) {3};
\node at (1.5,3.9) {4};
\node at (2,3.9) {5};
\node at (2.5,3.9) {6};
\node at (1,4.4) {Column};
\end{tikzpicture}
\caption{An example of a matrix $D(5,6)$}
\label{fig2}
\end{center}
\end{figure}

\subsection{Double-Scramble Shuffle}
A \textit{double-scramble shuffle} is an extension of a \textit{pile-scramble shuffle} developed by Ishikawa et al. \cite{scramble}. In the pile-scramble shuffle, we rearrange only the columns of the matrix by a random permutation; in the double-scramble shuffle, we rearrange both the selected rows and selected columns of the matrix by random permutations.

The formal steps of the double-scramble shuffle on a matrix $D(a,b)$ are as follows.
\begin{enumerate}
	\item Rearrange Rows $2,3,...,a$ by a uniformly random permutation $p = (p_2,p_3,...,p_a)$ of $(2,3,...,a)$, i.e. move Row $i$ to Row $p_i$ for every $i=2,3,...,a$ ($p$ is hidden from all parties).
	\item Rearrange Columns $1,2,...,b$ by a uniformly random permutation $q = (q_1,q_2,...,q_b)$ of $(1,2,...,b)$, i.e. move Column $j$ to Column $q_j$ for every $j=1,2,...,b$ ($q$ is hidden from all parties).
\end{enumerate}

Observe that the double-scramble shuffle hides the order of $x_2,x_3,...,x_a$ and the actual value of each $x_i$, but still preserves the number of rows that encode the same value as Row 1.

One of the possible ways to perform the double-scramble shuffle in real world is by the following procedures. In Step 1, the prover $P$ publicly puts the cards in each row into an envelope. Then, $P$ and the verifier $V$ jointly scramble the envelopes into a random permutation $p$ (which is unknown to both). Finally, $P$ publicly opens each envelope and put the cards in it back into a corresponding row. By doing this, $V$ can ensure that $P$ has only made row-wise swaps (and not arbitrary exchanges of cards). The same goes for column-wise swaps in Step 2.

\begin{remark}
This real-world implementation is based on an assumption that jointly scrambling the envelopes acts like an oracle that shuffles them into a uniformly random permutation unknown to all parties. This is also a sufficient assumption for the whole protocol to work correctly since all the randomness in our main protocol comes from double-scramble shuffles.
\end{remark}

\subsection{Rearrangement Protocol}
The sole purpose of a rearrangement protocol is to revert the cards (after we rearrange them) back to their original positions so that we can reuse the cards without revealing them. This protocol was implicitly used in some previous work on card-based protocols \cite{makaro,revert1,revert2,sudoku}.

The formal steps of the rearrangement protocol on a matrix $D(a,b)$ are as follows.
\begin{enumerate}
	\item Apply the double-scramble shuffle to the matrix.
	\item Publicly turn over all marking cards in Column 0. Suppose the opened cards are $p_2,p_3,...,p_a$ from top to bottom in this order.
	\item Publicly rearrange Rows $2,3,...,a$ by a permutation $p = (p_2,p_3,...,p_a)$, i.e. move Row $i$ to Row $p_i$ for every $i=2,3,...,a$.
	\item Publicly turn over all marking cards in Row 0. Suppose the opened cards are $q_1,q_2,...,q_b$ from left to right in this order.
	\item Publicly rearrange Columns $1,2,...,b$ by a permutation $q = (q_1,q_2,...,q_b)$, i.e. move Column $j$ to Column $q_j$ for every $j=1,2,...,b$.
	\item Publicly turn over all face-up cards.
\end{enumerate}

Note that since we first apply the double-scramble shuffle at Step 1, the order of Rows $2,3,...,a$ and the order of Columns $1,2,...,b$ are uniformly distributed among all possible permutations. Therefore, revealing marking cards in Steps 2 and 4 does not leak any information about the cards.

\section{Our Main Protocol}
\subsection{Well-Designed Puzzles} \label{well}
For simplicity, we first consider a special case of well-designed puzzles where it is more straightforward to construct a protocol.

Recall that a Numberlink grid has size $m \times n$, and has $k$ pairs of numbers $1,2,...,k$ written on some of its cells. In the solution, for each path joining two terminal cells with a number $x$, we put a number $x$ on every cell on that path (see Fig. \ref{fig3}). Since the puzzle is well-designed, every cell has a number on it and every path is simple.

\begin{figure}[H]
\begin{center}
\begin{tikzpicture}
\fill[lightgray] (3.2,3.2) rectangle (4.0,4.0);
\fill[lightgray] (0.0,2.4) rectangle (0.8,3.2);
\fill[lightgray] (0.8,2.4) rectangle (1.6,3.2);
\fill[lightgray] (3.2,2.4) rectangle (4.0,3.2);
\fill[lightgray] (0.0,1.6) rectangle (0.8,2.4);
\fill[lightgray] (2.4,1.6) rectangle (3.2,2.4);
\fill[lightgray] (0.0,0.0) rectangle (0.8,0.8);
\fill[lightgray] (3.2,0.0) rectangle (4.0,0.8);

\draw[step=0.8cm,color=gray] (0,0) grid (4,4);

\node at (0.4,3.6) {3};
\node at (1.2,3.6) {3};
\node at (2.0,3.6) {3};
\node at (2.8,3.6) {4};
\node at (3.6,3.6) {4};
\node at (0.4,2.8) {3};
\node at (1.2,2.8) {1};
\node at (2.0,2.8) {3};
\node at (2.8,2.8) {4};
\node at (3.6,2.8) {3};
\node at (0.4,2.0) {2};
\node at (1.2,2.0) {1};
\node at (2.0,2.0) {3};
\node at (2.8,2.0) {4};
\node at (3.6,2.0) {3};
\node at (0.4,1.2) {2};
\node at (1.2,1.2) {1};
\node at (2.0,1.2) {3};
\node at (2.8,1.2) {3};
\node at (3.6,1.2) {3};
\node at (0.4,0.4) {2};
\node at (1.2,0.4) {1};
\node at (2.0,0.4) {1};
\node at (2.8,0.4) {1};
\node at (3.6,0.4) {1};
\end{tikzpicture}
\caption{The way we fill numbers on cells according to the solution of the puzzle in Fig. \ref{fig1}, with gray cells being terminal cells}
\label{fig3}
\end{center}
\end{figure}

The intuition of this protocol is that the prover $P$ will try to convince the verifier $V$ that
\begin{enumerate}
	\item every terminal cell has exactly one adjacent cell with the same number, and
	\item every non-terminal cell has exactly two adjacent cells with the same number.
\end{enumerate}

For each terminal cell with a number $x$, $P$ publicly puts a sequence of face-down cards $E_k(x)$ on it. Then, for each non-terminal cell with a number $x$, $P$ secretly puts a sequence of face-down cards $E_k(x)$ on it.

The formal steps of the verification phase for each terminal cell $c$ are as follows.
\begin{enumerate}
	\item Publicly construct a matrix of cards in the following way: put the sequence on $c$ into Row 1, then put the sequence on each adjacent cell to $c$ in any order into each of the next four (or three, or two, if $c$ is on the edge or at the corner) rows. Finally, put the marking cards to complete the matrix $D(5,k)$ (or $D(4,k)$, or $D(3,k)$, for the edge or corner case).
	\item Apply the double-scramble shuffle to the matrix.
	\item Publicly turn over all encoding cards in Row 1. Locate the position of a \mbox{\mybox{$\heartsuit$}.} Suppose it is at Column $j$.
	\item Publicly turn over all other encoding cards in Column $j$. If there is exactly one \mybox{$\heartsuit$} besides the one in Row 1, then the protocol continues; otherwise $V$ rejects and the protocol terminates.
	\item Publicly turn over all face-up cards and apply the rearrangement protocol to the matrix to revert the cards to their original positions. Finally, publicly put the cards back to their corresponding cells.
\end{enumerate}

The verification phase for each non-terminal cell works exactly the same as that for a terminal cell, except that in Step 4, $V$ verifies that there are exactly two (instead of one) \mybox{$\heartsuit$}s in Column $j$ besides the one in Row 1.

$P$ performs the verification phase for every cell in the grid\footnote{Step 5 is not necessary when verifying the last cell in the grid.}. If every cell passes the verification, then $V$ accepts.

In total, our protocol for a well-designed puzzle uses $kmn$ encoding cards and $k+4$ marking cards. Therefore, the total number of required cards is $\Theta(kmn)$.

\begin{remark}
In this protocol, $P$ can convince $V$ that he/she knows a solution, but cannot convince $V$ that the puzzle is well-designed or that all cells are covered by paths in his/her solution (see Fig. \ref{fig4}).
\end{remark}

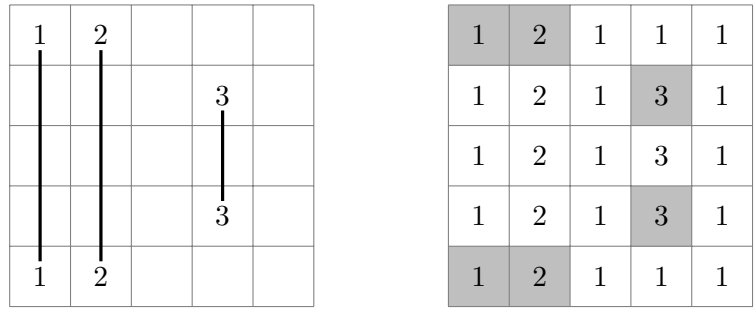
\begin{figure}[H]
\begin{center}
\begin{tikzpicture}
\draw[step=0.8cm,color=gray] (0,0) grid (4,4);

\node at (0.4,3.6) {1};
\node at (0.4,0.4) {1};
\node at (1.2,3.6) {2};
\node at (1.2,0.4) {2};
\node at (2.8,2.8) {3};
\node at (2.8,1.2) {3};

\draw[very thick] (0.4,3.4) -- (0.4,0.6);
\draw[very thick] (1.2,3.4) -- (1.2,0.6);
\draw[very thick] (2.8,2.6) -- (2.8,1.4);
\end{tikzpicture}
\hspace{1.5cm}
\begin{tikzpicture}
\fill[lightgray] (0.0,0.0) rectangle (0.8,0.8);
\fill[lightgray] (0.8,0.0) rectangle (1.6,0.8);
\fill[lightgray] (0.0,3.2) rectangle (0.8,4.0);
\fill[lightgray] (0.8,3.2) rectangle (1.6,4.0);
\fill[lightgray] (2.4,0.8) rectangle (3.2,1.6);
\fill[lightgray] (2.4,2.4) rectangle (3.2,3.2);

\draw[step=0.8cm,color=gray] (0,0) grid (4,4);

\node at (0.4,3.6) {1};
\node at (1.2,3.6) {2};
\node at (2.0,3.6) {1};
\node at (2.8,3.6) {1};
\node at (3.6,3.6) {1};
\node at (0.4,2.8) {1};
\node at (1.2,2.8) {2};
\node at (2.0,2.8) {1};
\node at (2.8,2.8) {3};
\node at (3.6,2.8) {1};
\node at (0.4,2.0) {1};
\node at (1.2,2.0) {2};
\node at (2.0,2.0) {1};
\node at (2.8,2.0) {3};
\node at (3.6,2.0) {1};
\node at (0.4,1.2) {1};
\node at (1.2,1.2) {2};
\node at (2.0,1.2) {1};
\node at (2.8,1.2) {3};
\node at (3.6,1.2) {1};
\node at (0.4,0.4) {1};
\node at (1.2,0.4) {2};
\node at (2.0,0.4) {1};
\node at (2.8,0.4) {1};
\node at (3.6,0.4) {1};
\end{tikzpicture}
\end{center}
\caption{In this example, a puzzle is not well-designed and the prover $P$ knows a solution that does not cover all cells (left). However, it is still possible for $P$ to fill numbers on cells in a way that will get accepted in the protocol (right).}
\label{fig4}
\end{figure}

\subsection{General Puzzles} \label{general}
Now we consider a general case where the puzzle may not be well-designed, and the paths in the prover's solution may not cover all cells. We will employ some additional tricks to the protocol in Section \ref{well} to make it support general puzzles as well.

First, if the solution contains a non-simple path $(c_1,c_2,...,c_t)$ with $c_i$ being adjacent to $c_j$ where $j>i+1$, then we replace it with a shorter path $(c_1,c_2,...,c_i,c_j,c_{j+1}...,c_t)$. We repeatedly perform this until every path in the solution becomes simple.

We put a number on each cell that is covered by a path the same way as the protocol in Section \ref{well}. For each cell $c$ in the $i$-th row and $j$-th column, we call $c$ an \textit{even cell} if $i+j$ is even, and an \textit{odd cell} if $i+j$ is odd. Then, we put a number $k+1$ on each even cell not covered by any path, and a number $k+2$ on each odd cell not covered by any path (see Fig. \ref{fig5}). Observe that by filling the numbers this way, each cell not covered by any path will have no adjacent cell with the same number.

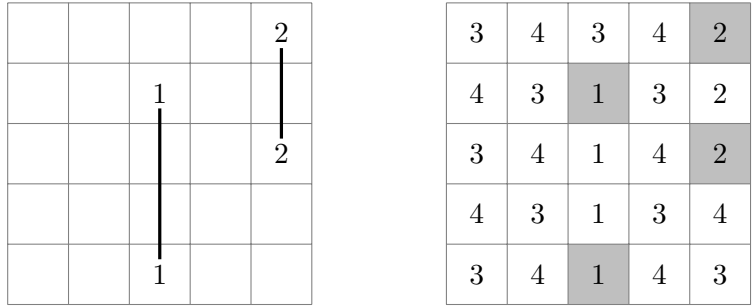
\begin{figure}[H]
\begin{center}
\begin{tikzpicture}
\draw[step=0.8cm,color=gray] (0,0) grid (4,4);

\node at (3.6,3.6) {2};
\node at (2.0,2.8) {1};
\node at (3.6,2.0) {2};
\node at (2.0,0.4) {1};

\draw[very thick] (3.6,3.4) -- (3.6,2.2);
\draw[very thick] (2.0,2.6) -- (2.0,0.6);
\end{tikzpicture}
\hspace{1.5cm}
\begin{tikzpicture}
\fill[lightgray] (3.2,3.2) rectangle (4.0,4.0);
\fill[lightgray] (1.6,2.4) rectangle (2.4,3.2);
\fill[lightgray] (3.2,1.6) rectangle (4.0,2.4);
\fill[lightgray] (1.6,0.0) rectangle (2.4,0.8);

\draw[step=0.8cm,color=gray] (0,0) grid (4,4);

\node at (0.4,3.6) {3};
\node at (1.2,3.6) {4};
\node at (2.0,3.6) {3};
\node at (2.8,3.6) {4};
\node at (3.6,3.6) {2};
\node at (0.4,2.8) {4};
\node at (1.2,2.8) {3};
\node at (2.0,2.8) {1};
\node at (2.8,2.8) {3};
\node at (3.6,2.8) {2};
\node at (0.4,2.0) {3};
\node at (1.2,2.0) {4};
\node at (2.0,2.0) {1};
\node at (2.8,2.0) {4};
\node at (3.6,2.0) {2};
\node at (0.4,1.2) {4};
\node at (1.2,1.2) {3};
\node at (2.0,1.2) {1};
\node at (2.8,1.2) {3};
\node at (3.6,1.2) {4};
\node at (0.4,0.4) {3};
\node at (1.2,0.4) {4};
\node at (2.0,0.4) {1};
\node at (2.8,0.4) {4};
\node at (3.6,0.4) {3};
\end{tikzpicture}
\end{center}
\caption{An example of a solution of a puzzle that is not well-designed (left), and the way we put numbers on the grid (right)}
\label{fig5}
\end{figure}

The intuition of this protocol is that the prover $P$ will try to convince the verifier $V$ that
\begin{enumerate}
	\item every terminal cell has exactly one adjacent cell with the same number, and
	\item every non-terminal cell either has a number greater than $k$, or has exactly two adjacent cells with the same number.
\end{enumerate}

Since the maximum number on the grid is at most $k+2$, we always use $E_{k+2}(x)$ instead of $E_k(x)$ to encode a number $x$ in this protocol. For each terminal cell with a number $x$, $P$ publicly puts a sequence of face-down cards $E_{k+2}(x)$ on it. Then, for each non-terminal cell with a number $x$, $P$ secretly puts a sequence of face-down cards $E_{k+2}(x)$ on it.

For each terminal cell, the verification phase works exactly the same as the protocol in Section \ref{well} (except the size of the matrix will be $D(a,k+2)$ instead of $D(a,k)$ for $a \in \{3,4,5\}$). For each non-terminal even (resp. odd) cell, we put two additional rows, both encoding the number $k+1$ (resp. $k+2$), to the bottom of the matrix. The formal steps for verifying each non-terminal cell $c$ are as follows.
\begin{enumerate}
	\item Publicly construct a matrix of cards in the following way: put the sequence on $c$ into Row 1, then put the sequence on each adjacent cell to $c$ into each of the next four (or three, or two, if $c$ is on the edge or at the corner) rows in any order. Then, if $c$ is an even cell (resp. odd cell), put two copies of a sequence $E_{k+2}(k+1)$ (resp. $E_{k+2}(k+2)$) into the next two rows. Finally, put the marking cards to complete the matrix $D(7,k+2)$ (or $D(6,k+2)$, or $D(5,k+2)$, for the edge or corner case).
	\item Apply the double-scramble shuffle to the matrix.
	\item Publicly turn over all encoding cards in Row 1. Locate the position of a \mbox{\mybox{$\heartsuit$}.} Suppose it is at Column $j$.
	\item Publicly turn over all other encoding cards in Column $j$. If there are exactly two \mbox{\mybox{$\heartsuit$}s} besides the one in Row 1, then the protocol continues; otherwise $V$ rejects and the protocol terminates.
	\item Publicly turn over all face-up cards and apply the rearrangement protocol to the matrix to revert the cards to their original positions. Finally, publicly put the cards back to their corresponding cells.
\end{enumerate}

In total, our protocol for a general puzzle uses $(k+2)(mn+2)$ encoding cards and $k+8$ marking cards. Therefore, the total number of required cards is still $\Theta(kmn)$.

\section{Proof of Correctness and Security}
We will prove the perfect completeness, perfect soundness, and zero-knowledge properties of the protocol for a general puzzle in Section \ref{general}.

\begin{lemma}[Perfect Completeness] \label{lem1}
If $P$ knows a solution of the Numberlink puzzle, then $V$ always accepts.
\end{lemma}

\begin{proof}
Suppose that $P$ knows a solution that contains only simple paths, and fills numbers on the grid according to that solution.

\begin{itemize}
	\item Consider each terminal cell $c$ with a number $x \leq k$. There must be a path $(c_1,c_2,...,c_t)$ starting at $c_1=c$ and ending at $c_t$, the other terminal cell with the number $x$. Since each cell in the grid either belongs to some path or has a number $k+1$ or $k+2$ on it, the set of all cells having the number $x$ is exactly $\{c_1,c_2,...,c_t\}$. We know that $c_2$ is adjacent to $c$ and has the number $x$. Moreover, since the path is simple, there cannot be an index $i>2$ such that $c_i$ is adjacent to $c$. Therefore, $c$ has exactly one adjacent cell with the same number. Since the double-scramble shuffle preserves the number of rows that encode a value equal to that of Row 1, the verification phase for $c$ will pass.
	
	\item Consider each non-terminal cell $c$ with a number $x \leq k$. There must be a path $(c_1,c_2,...,c_t)$ joining two terminal cells with the number $x$. As previously shown, the set of all cells having the number $x$ is exactly $\{c_1,c_2,...,c_t\}$, so we have $c=c_i$ for some index $i$ where $1<i<t$. We know that $c_{i-1}$ and $c_{i+1}$ are adjacent to $c$ and have the number $x$. Moreover, since the path is simple, there cannot be an index $j$ other than $i-1$ and $i+1$ such that $c_j$ is adjacent to $c$. Therefore, $c$ has exactly two adjacent cells with the same number. Since the double-scramble shuffle preserves the number of rows that encode a value equal to that of Row 1, the verification phase for $c$ will pass.
	
	\item Consider each non-terminal cell $c$ with a number $x = k+1$ or $k+2$. Recall that by the way we put numbers on the cells, $c$ has no adjacent cell with the same number. However, in the verification phase of $c$, we put two additional rows, both encoding the number $x$, to the matrix. Therefore, there will be exactly two rows that encode a value equal to that of Row 1, hence the verification phase for $c$ will pass.
\end{itemize}

Since the verification phase for every cell passes, $V$ will always accept.
\end{proof}

\begin{lemma}[Perfect Soundness] \label{lem2}
If $P$ does not know a solution of the Numberlink puzzle, then $V$ always rejects.
\end{lemma}

\begin{proof}
We will prove the contrapositive of this statement. Suppose that $V$ accepts, meaning that the verification phase passes for every cell. We will prove that $P$ must know a solution.

Consider any number $x \leq k$. We know that there are two terminal cells with the number $x$. Consider one of them, called $c_1$. We know from the verification phase that $c_1$ has exactly one adjacent cell with the number $x$, called $c_2$. For each $i \geq 2$, if $c_i$ is a terminal cell, then there exists a path $(c_1,c_2,...,c_i)$ connecting the two terminal cells with the number $x$. Otherwise if $c_i$ is a non-terminal cell, then we know from the verification phase that $c_i$ has exactly two adjacent cells with the number $x$, one of them being $c_{i-1}$. We then inductively proceed to consider the other cell, called $c_{i+1}$, in the same manner. Also, $c_{i+1}$ must be different from every $c_j$ with $j \leq i$. Therefore, we must eventually reach the other terminal cell, which implies that there exists a path connecting the two terminal cells with the number $x$.

Since this is true for every number $x \leq k$, there exists a set of disjoint paths joining all pairs of terminal cells with the same number in $P$'s solution, implying that $P$ must know a solution.
\end{proof}

\begin{lemma}[Zero-Knowledge] \label{lem3}
During the verification phase, $V$ learns nothing about $P$'s solution of the Numberlink puzzle.
\end{lemma}

\begin{proof}
To prove the zero-knowledge property, it is sufficient to prove that all distributions of the values that appear when $P$ turns over cards can be simulated by a simulator $S$ without knowing $P$'s solution.

Consider the verification phase of a cell $c$ with a matrix $D(a,k+2)$ of cards ($a \in \{3,4,5\}$ for a terminal cell and $a \in \{5,6,7\}$ for a non-terminal cell). There are two steps in the verification phase where $P$ turns over cards.

In the step where $P$ turns over all encoding cards in Row 1, the order of Columns $1,2,...,k+2$ is uniformly distributed among all possible permutations due to the double-scramble shuffle, hence the \mybox{$\heartsuit$} has an equal probability to appear at each of the $k+2$ positions. Therefore, this step can be simulated by $S$ without knowing $P$'s solution.

After that, $P$ locates the position of the \mybox{$\heartsuit$} in Row 1 to be at Column $j$, and then turns over all other encoding cards in Column $j$. The order of Rows $2,3,...,a$ is uniformly distributed among all possible permutations due to the double-scramble shuffle, hence all (one or two) \mybox{$\heartsuit$}s have an equal probability to appear at each of the $a-1$ positions. Therefore, this step can be simulated by $S$ without knowing $P$'s solution.

Therefore, we can conclude that $V$ learns nothing about $P$'s solution during the verification phase.
\end{proof}

\section{Applications}
\subsection{Undirected $k$ Vertex-Disjoint Paths Problem} \label{ukdpp}
Consider the following problem: given an undirected graph $G$, and $k$ pairs of vertices $(s_1,t_1), ..., (s_k,t_k)$ called \textit{terminal vertices}, find a set of vertex-disjoint paths joining every pair of $s_i$ and $t_i$, or report that none exists. This problem is called the undirected $k$ vertex-disjoint paths problem (\textsc{u}$k$-\textsc{dpp}) and is one of the most well-studied problems in algorithmic graph theory. \textsc{u}$k$-\textsc{dpp} is known to be solvable in polynomial time for a fixed constant $k$ \cite{robertson}, but becomes NP-complete when $k$ is a part of the input \cite{karp}.

The Numberlink puzzle can be considered as a special case of \textsc{u}$k$-\textsc{dpp} in a grid graph. In this section, we will modify the protocol for Numberlink in Section \ref{general} to make it support \textsc{u}$k$-\textsc{dpp} in a general graph as well.

In \textsc{u}$k$-\textsc{dpp}, a path $(v_1,v_2,...,v_t)$ is called \textit{simple} if there is no $i,j$ such that $j>i+1$ and $v_j$ is a neighbor of $v_i$. Similarly to the protocol for Numberlink, if our solution contains a non-simple path $(v_1,v_2,...,v_t)$ with $v_j$ being a neighbor of $v_i$ where $j>i+1$, then we replace it with a shorter path $(v_1,v_2,...,v_i,v_j,v_{j+1}...,v_t)$. We repeatedly perform this until every path in the solution becomes simple.

Suppose that the maximum degree of a vertex in $G$ is $d$. We can inductively color the vertices of $G$ with at most $d+1$ colors in linear time such that there are no neighboring vertices with the same color. This $(d+1)$-coloring is known to all parties. Similarly to the protocol for Numberlink, for each path connecting $s_x$ and $t_x$, we put a number $x$ on every vertex on that path. For each vertex $v$ not covered by any path, we put a number $k+i$ on $v$ if it has the $i$-th color in the $(d+1)$-coloring of $G$. By filling the numbers this way, each vertex not covered by any path will have no neighbor with the same number.

Let $T$ be the set of all terminal vertices. The intuition of this protocol is that the prover $P$ will try to convince the verifier $V$ that
\begin{enumerate}
	\item every vertex in $T$ has exactly one neighbor with the same number, and
	\item every vertex not in $T$ either has a number greater than $k$, or has exactly two neighbors with the same number.
\end{enumerate}

Since the maximum number on the vertices is at most $k+d+1$, we use $E_{k+d+1}(x)$ to encode a number $x$. For each vertex $v \in T$ with a number $x$, $P$ publicly puts a sequence of face-down cards $E_{k+d+1}(x)$ on $v$. Then, for each vertex $v \notin T$ with a number $x$, $P$ secretly puts a sequence of face-down cards $E_{k+d+1}(x)$ on $v$.

The verification phase works in the same manner as the protocol for Numberlink. For a vertex $v \in T$, $P$ puts the sequence on $v$ into the first row, and the sequence on each of $v$'s neighbors into each of the next (at most) $d$ rows of the matrix. $V$ then verifies that there is exactly one \mybox{$\heartsuit$} in the same column as the \mybox{$\heartsuit$} in Row 1. For a vertex $v \notin T$ with the $i$-th color in the $(d+1)$-coloring of $G$, $P$ does the same but also puts two additional rows, both encoding $k+i$, to the matrix. Then, $V$ verifies that there are exactly two \mybox{$\heartsuit$}s in the same column as the \mybox{$\heartsuit$} in Row 1.

The proofs of Lemmas \ref{lem1}, \ref{lem2}, and \ref{lem3} can be applied straightforwardly to show the correctness and security of this protocol. In total, this protocol uses $(k+d+1)(|V_G|+2)$ encoding cards and $k+2d+3$ marking cards, where $V_G$ is the set of vertices of $G$. Therefore, the total number of required cards is $\Theta((k+d)|V_G|)$.

\subsection{Directed $k$ Vertex-Disjoint Paths Problem} \label{dkdpp}
The directed $k$ vertex-disjoint paths problem (\textsc{d}$k$-\textsc{dpp}) is a counterpart of \textsc{u}$k$-\textsc{dpp} in a directed graph, with an objective to find a set of vertex-disjoint directed paths from every $s_i$ to $t_i$. \textsc{d}$k$-\textsc{dpp} is significantly harder than \textsc{u}$k$-\textsc{dpp}, as the problem is NP-complete even for any constant $k \geq 2$ \cite{fortune}. Our protocol for \textsc{u}$k$-\textsc{dpp} in Section \ref{dkdpp} can be slightly modified as follows to support \textsc{d}$k$-\textsc{dpp}.

First, we make all paths in the solution become simple and fill numbers on vertices of $G$ in exactly the same way as the protocol for \textsc{u}$k$-\textsc{dpp} ($d$ is still defined to be the maximum degree (sum of indegree and outdegree) of a vertex in $G$).

Let $T_s = \{s_1,s_2,...,s_k\}$ be the set of \textit{source vertices} and $T_t = \{t_1,t_2,...,t_k\}$ be the set of \textit{sink vertices}. The intuition of this protocol is that the prover $P$ will try to convince the verifier $V$ that
\begin{enumerate}
	\item every vertex in $T_s$ has no incoming neighbor with the same number and exactly one outgoing neighbor with the same number, and
	\item every vertex in $T_t$ has exactly one incoming neighbor with the same number and no outgoing neighbor with the same number, and
	\item every vertex not in $T_s \cup T_t$ either has a number greater than $k$, or has exactly one incoming neighbor with the same number and exactly one outgoing neighbor with the same number.
\end{enumerate}

Like in the protocol for \textsc{u}$k$-\textsc{dpp}, for each vertex $v \in T$ with a number $x$, $P$ publicly puts a sequence of face-down cards $E_{k+d+1}(x)$ on $v$. Then, for each vertex $v \notin T$ with a number $x$, $P$ secretly puts a sequence of face-down cards $E_{k+d+1}(x)$ on $v$.

The verification phase works in the same manner as the protocol for \textsc{u}$k$-\textsc{dpp}, but each vertex has to be verified in two separate rounds, one for incoming neighbors and one for outgoing neighbors. In the incoming (resp. outgoing) round of a vertex $v \in T_s$, $P$ puts the sequence on $v$ into the first row, and the sequence on each of $v$'s incoming (resp. outgoing) neighbors into each of the next at most $d$ rows of the matrix. Then, $V$ verifies that there are exactly zero (resp. one) \mbox{\mybox{$\heartsuit$}s} in the same column as the \mybox{$\heartsuit$} in Row 1. The other way around goes for each vertex $v \in T_t$. In each round of a vertex $v \notin T_s \cup T_t$ with the $i$-th color in the $(d+1)$-coloring of $G$, $P$ does the same but also puts an additional row encoding $k+i$ to the matrix. Then, $V$ verifies that there is exactly one \mybox{$\heartsuit$} in the same column as the \mybox{$\heartsuit$} in Row 1.

The proofs of Lemmas \ref{lem1}, \ref{lem2}, and \ref{lem3} can be applied straightforwardly to show the correctness and security of this protocol. In total, this protocol uses $(k+d+1)(|V_G|+1)$ encoding cards and $k+2d+2$ marking cards. Therefore, the total number of required cards is $\Theta((k+d)|V_G|)$, the same as in \textsc{u}$k$-\textsc{dpp}.

\section{Future Work}
We developed a physical protocol of zero-knowledge proof for Numberlink puzzle using $\Theta(kmn)$ cards, and for \textsc{u}$k$-\textsc{dpp} and \textsc{d}$k$-\textsc{dpp} using $\Theta((k+d)|V_G|)$ cards. A challenging future work is to develop a protocol of zero-knowledge proof for Numberlink puzzle that requires asymptotically fewer number of cards, or the one that can convince the verifier that the prover's solution contains paths that cover all cells (which is apparently a requirement in a variant of rule used in some newly developed mobile apps).

Other possible future work includes developing protocols of zero-knowledge proof for other popular logic puzzles or well-known problems in algorithmic graph theory, as well as exploring methods to physically verify other interesting functions.

\end{document}